\documentclass[a4paper,12pt,reqno]{amsart}

\usepackage{graphicx,amssymb,datetime,float,MnSymbol,currfile,tikz}
\usepackage[round]{natbib}
\usetikzlibrary{arrows}
\usetikzlibrary{decorations.markings}

\newtheorem{theorem}{Theorem}

\newtheorem{corollary}[theorem]{Corollary}

\newcommand{\comments}[1]{}

\def\na{\overline{a}}
\def\nc{\overline{c}}
\def\nd{\overline{d}}

\tikzset{tt/.style={decoration={
  markings,
  mark=at position .485 with {\arrow{>}},
  mark=at position .515 with {\arrow{<}}},postaction={decorate}}}

\begin{document}

\title[]{On the Non-Monotonicity of a Non-Differentially Mismeasured Binary Confounder}

\author[]{Jose M. Pe\~{n}a\\
IDA, Link\"oping University, Sweden\\
jose.m.pena@liu.se}


\maketitle

\begin{abstract}
Suppose that we are interested in the average causal effect of a binary treatment on an outcome when this relationship is confounded by a binary confounder. Suppose that the confounder is unobserved but a non-differential binary proxy of it is observed. We identify conditions under which adjusting for the proxy comes closer to the incomputable true average causal effect than not adjusting at all. Unlike other works, we do not assume that the average causal effect of the confounder on the outcome is in the same direction among treated and untreated.
\end{abstract}

\begin{figure}[t]
\begin{tabular}{c|c}
\begin{tikzpicture}[inner sep=1mm]
\node at (0,0) (A) {$A$};
\node at (1,.5) (D) {$D$};
\node at (2,0) (Y) {$Y$};
\node at (1,1.5) (C) {$C$};
\path[->] (A) edge (Y);
\path[->] (C) edge (A);
\path[->] (C) edge (Y);
\path[->] (C) edge (D);
\end{tikzpicture}
&
\begin{tikzpicture}[inner sep=1mm]
\node at (0,0) (X) {$A$};
\node at (2.5,1.5) (Z) {$D$};
\node at (2,0) (Y) {$Y$};
\node at (1,1.5) (U) {$C$};
\path[->] (X) edge node[above] {$\alpha$} (Y);
\path[->] (U) edge node[left] {$\beta$} (X);
\path[->] (U) edge node[right] {$\gamma$} (Y);
\path[->] (U) edge node[above] {$\delta$} (Z);
\end{tikzpicture}
\end{tabular}\caption{Left: Causal graph where $Y$ is a discrete or continuous random variable, and $A$, $C$ and $D$ are binary random variables. Moreover, $C$ is unobserved. Right: Path diagram where $C$ is unobserved.}\label{fig:graphs}
\end{figure}
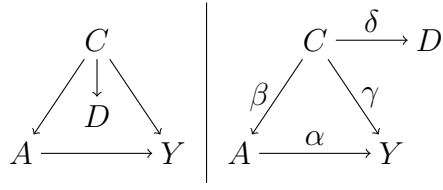

\section{Introduction}

Suppose that we are interested in the average causal effect of a binary treatment $A$ on an outcome $Y$ when this relationship is confounded by a binary confounder $C$. Suppose also that $C$ is non-differentially mismeasured, meaning that (i) $C$ is not observed and, instead, a binary proxy $D$ of $C$ is observed, and (ii) $D$ is conditionally independent of $A$ and $Y$ given $C$. The causal graph to the left in Figure \ref{fig:graphs} represents the relationships between the random variables.

\citet{Greenland1980} argues that adjusting for $D$ produces a partially adjusted measure of the average causal effect of $A$ on $Y$ that is between the crude (i.e., unadjusted) and the true (i.e., adjusted for $C$) measures and, thus, it comes closer to the incomputable true measure than the crude one. \citet{OgburnandVanderWeele2012a} show that, although this result does not always hold, it does hold under some monotonicity condition in $C$. Specifically, $E[Y|A,C]$ must be non-decreasing or non-increasing in $C$. Unfortunately, the condition cannot be verified empirically because $C$ is unobserved. \cite{OgburnandVanderWeele2013} extend these results to the case where $C$ takes more than two values. \citet{Penna2020} shows that if $E[Y|A,D]$ is non-decreasing or non-increasing in $D$ (which can be verified empirically), then so is $E[Y|A,C]$ with respect to $C$ and, thus, the partially adjusted average causal effect lies between the crude and the true ones. Finally, if there are at least two independent proxies of $C$, then \citet{Miaoetal.2018} show that the average causal effect of $A$ on $Y$ can be identified under certain rank condition.

In this paper, we focus on the case where neither $E[Y|A,C]$ nor $E[Y|A,D]$ are monotone in $C$ or $D$. And we report conditions under which the partially adjusted average causal effect is still between the crude and the true ones and, thus, it is still closer to the incomputable true average causal effect. Specifically, the rest of the paper is organized as follows. Sections \ref{sec:nonmonotone} and \ref{sec:nonmonotone2} report the novel conditions. Section \ref{sec:pathdiagrams} deals with continuous random variables. Section \ref{sec:discussion} closes with some discussion.

\section{Bounding the Observed Risk Difference}\label{sec:nonmonotone}

Consider the causal graph to the left in Figure \ref{fig:graphs}, where $Y$ is a discrete or continuous random variable, and $A$, $C$ and $D$ are binary random variables. The graph entails the following factorization:
\begin{equation}\label{eq:factorization}
p(A,C,D,Y)=p(C)p(D|C)p(A|C)p(Y|A,C).
\end{equation}
Let $A$ take values $a$ and $\na$, and similarly for $C$ and $D$. Let $A$, $D$ and $Y$ be observed and let $C$ be unobserved. Let $Y_a$ and $Y_{\na}$ denote the counterfactual outcomes under treatments $A=a$ and $A=\na$, respectively. The average causal effect of $A$ on $Y$ or true risk difference ($RD_{true}$) is defined as $RD_{true}=E[Y_a]-E[Y_{\na}]$. It can be rewritten as follows \cite[Theorem 3.3.2]{Pearl2009}:
\[
RD_{true} = E[Y|a,c]p(c) + E[Y|a,\nc]p(\nc) - E[Y|\na,c]p(c) - E[Y|\na,\nc]p(\nc).
\]
Since $C$ is unobserved, $RD_{true}$ cannot be computed. However, it can be approximated by the unadjusted average causal effect or crude risk difference ($RD_{crude}$):
\[
RD_{crude} = E[Y|a] - E[Y|\na]
\]
and by the partially adjusted average causal effect or observed risk difference ($RD_{obs}$):
\[
RD_{obs} = E[Y|a,d]p(d) + E[Y|a,\nd]p(\nd) - E[Y|\na,d]p(d) - E[Y|\na,\nd]p(\nd).
\]
Now the question is, which of the two approximations comes closer to the true quantity ? This paper aims to answer this question.

We say that $E[Y|A,C]$ is non-decreasing in $C$ if 
\[
E[Y|a, c] \geq E[Y|a, \nc] \text{ and } E[Y|\na, c] \geq E[Y|\na, \nc].
\]
Likewise, $E[Y|A,C]$ is non-increasing in $C$ if
\[
E[Y|a, c] \leq E[Y|a, \nc] \text{ and } E[Y|\na, c] \leq E[Y|\na, \nc].
\]
Moreover, $E[Y|A,C]$ is monotone in $C$ if it is non-decreasing or non-increasing in $C$, i.e. the average causal effect of $C$ on $Y$ is in the same direction among the treated ($A=a$) and the untreated ($A=\na$). \citet[Result 1]{OgburnandVanderWeele2012a} show that if $E[Y|A,C]$ is monotone in $C$, then $RD_{obs}$ lies between $RD_{true}$ and $RD_{crude}$ and, thus, it comes closer to $RD_{true}$ than $RD_{crude}$. Unfortunately, the antecedent of this rule cannot be verified empirically, because $C$ is unobserved. Therefore, one must rely on substantive knowledge to apply the rule. \citet[Corollay 2]{Penna2020} shows that if $E[Y|A,D]$ is monotone in $D$, then $RD_{obs}$ lies between $RD_{true}$ and $RD_{crude}$. Note that the antecedent of this rule can be verified empirically. Actually, $E[Y|A,C]$ is monotone in $C$ if and only if $E[Y|A,D]$ is monotone in $D$ \citep{OgburnandVanderWeele2012a,Penna2020}.

\citet[Theorems 3 and 4]{Penna2020} characterizes a case where $E[Y|A,C]$ is not monotone in $C$ and, thus, $E[Y|A,D]$ is not monotone in $D$, and yet $RD_{obs}$ lies between $RD_{true}$ and $RD_{crude}$. We re-state this result in the next theorem. Note that one must rely on substantive knowledge to verify the conditions in the theorem.

\begin{theorem}[Pe\~na, 2020, Theorems 3 and 4]\label{the:Penna2020}
Consider the causal graph to the left in Figure \ref{fig:graphs}. Let $p(c)=0.5$ and $p(a|c)=p(\na|\nc)=p(d|c)=p(\nd|\nc) \geq 0.5$. If $E[Y|a,c] - E[Y|a,\nc] \geq E[Y|\na,\nc] - E[Y|\na,c] \geq 0$, then $RD_{crude} \geq RD_{obs} \geq RD_{true}$. If $E[Y|a,c] - E[Y|a,\nc] \leq E[Y|\na,\nc] - E[Y|\na,c] \leq 0$, then $RD_{crude} \leq RD_{obs} \leq RD_{true}$.
\end{theorem}

The following theorems are the main contribution of this work. They show that the conditions in the previous theorem can be relaxed. Their proofs can be found in the supplementary material.

\begin{theorem}\label{the:nonmonotone}
Consider the causal graph to the left in Figure \ref{fig:graphs}. Let $p(c)=0.5$, $p(a|c)=p(\na|\nc) \geq 0.5$ and $p(d|c)=p(\nd|\nc) \geq 0.5$. Then, $RD_{obs}$ lies between $RD_{true}$ and $RD_{crude}$.
\end{theorem}

\begin{theorem}\label{the:nonmonotone2}
Consider the causal graph to the left in Figure \ref{fig:graphs}. Let $p(c)=0.5$, $p(a|c)=p(\na|\nc) \leq 0.5$ and $p(d|c)=p(\nd|\nc) \leq 0.5$. Then, $RD_{obs}$ lies between $RD_{true}$ and $RD_{crude}$.
\end{theorem}

The following example gives some intuition about the conditions in Theorem \ref{the:nonmonotone}. Let $A$, $D$ and $Y$ represent three diseases, and $C$ a gene variant that affects the three of them. Moreover, suppose that suffering $A$ affects the risk of suffering $Y$. Suppose also that half of the population carries the gene variant $C$, i.e. $p(c)=0.5$. Suppose also that carrying $C$ predisposes to suffer $A$ and $D$ as much as not carrying it protects against the diseases, i.e. $p(a|c)=p(\na|\nc) \geq 0.5$ and $p(d|c)=p(\nd|\nc) \geq 0.5$. Then, the theorem applies.

\begin{corollary}\label{cor:nonmonotoneb}
Consider the causal graph to the left in Figure \ref{fig:graphs}. Let $p(c)=0.5$, $p(a|c)=p(\na|\nc) \geq 0.5$ and $p(d|c)=p(\nd|\nc) \geq 0.5$. If $E[Y|a,c] - E[Y|a,\nc] \geq E[Y|\na,\nc] - E[Y|\na,c] \geq 0$, then $RD_{crude} \geq RD_{obs} \geq RD_{true}$. If $E[Y|a,c] - E[Y|a,\nc] \leq E[Y|\na,\nc] - E[Y|\na,c] \leq 0$, then $RD_{crude} \leq RD_{obs} \leq RD_{true}$.
\end{corollary}

\begin{corollary}\label{cor:nonmonotone2b}
Consider the causal graph to the left in Figure \ref{fig:graphs}. Let $p(c)=0.5$, $p(a|c)=p(\na|\nc) \leq 0.5$ and $p(d|c)=p(\nd|\nc) \leq 0.5$. If $E[Y|a,c] - E[Y|a,\nc] \geq E[Y|\na,\nc] - E[Y|\na,c] \geq 0$, then $RD_{crude} \leq RD_{obs} \leq RD_{true}$. If $E[Y|a,c] - E[Y|a,\nc] \leq E[Y|\na,\nc] - E[Y|\na,c] \leq 0$, then $RD_{crude} \geq RD_{obs} \geq RD_{true}$.
\end{corollary}

To get some intuition about the conditions for the first result in Corollary \ref{cor:nonmonotoneb}, let us extend the previous example with the following additional assumption: Carrying the gene variant $C$ increases the average severity of $Y$ for the individuals suffering $A$ more than it decreases the severity for the rest. Then, the corollary applies.

Note that one must rely on substantive knowledge to verify the conditions in the previous theorems and corollaries. The next two corollaries show that this can partially be alleviated by replacing the conditions on $E[Y|A,C]$ with similar conditions on $E[Y|A,D]$: The former are not empirically testable because $C$ is unobserved, but the latter are.

\begin{corollary}\label{cor:nonmonotonec}
Under the conditions in Corollary \ref{cor:nonmonotoneb}, $E[Y|a,c] - E[Y|a,\nc] \geq E[Y|\na,\nc] - E[Y|\na,c] \geq 0$ if and only if $E[Y|a,d] - E[Y|a,\nd] \geq E[Y|\na,\nd] - E[Y|\na,d] \geq 0$. Likewise when replacing $\geq$ with $\leq$.
\end{corollary}

\begin{corollary}\label{cor:nonmonotone2c}
Under the conditions in Corollary \ref{cor:nonmonotone2b}, $E[Y|a,c] - E[Y|a,\nc] \geq E[Y|\na,\nc] - E[Y|\na,c] \geq 0$ if and only if $E[Y|a,d] - E[Y|a,\nd] \leq E[Y|\na,\nd] - E[Y|\na,d] \leq 0$. Likewise when swapping $\leq$ and $\geq$.
\end{corollary}

\begin{figure}[t]
\begin{tabular}{cc}
\includegraphics[scale=.38]{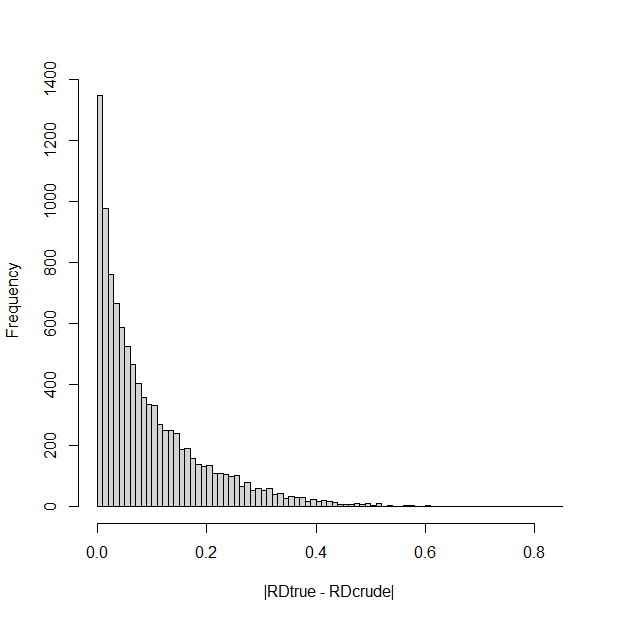}&\includegraphics[scale=.38]{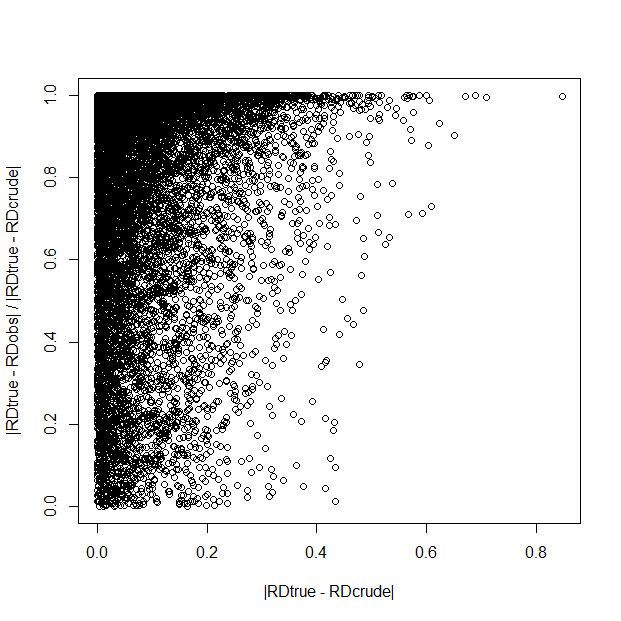}\\
\includegraphics[scale=.38]{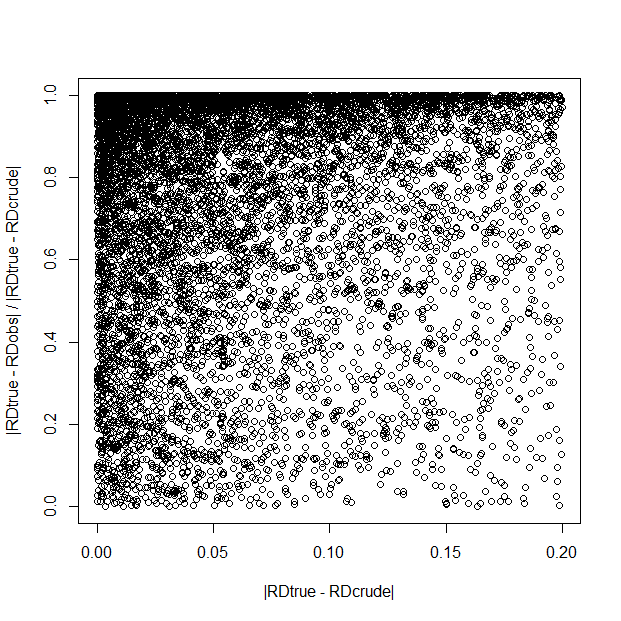}&\includegraphics[scale=.38]{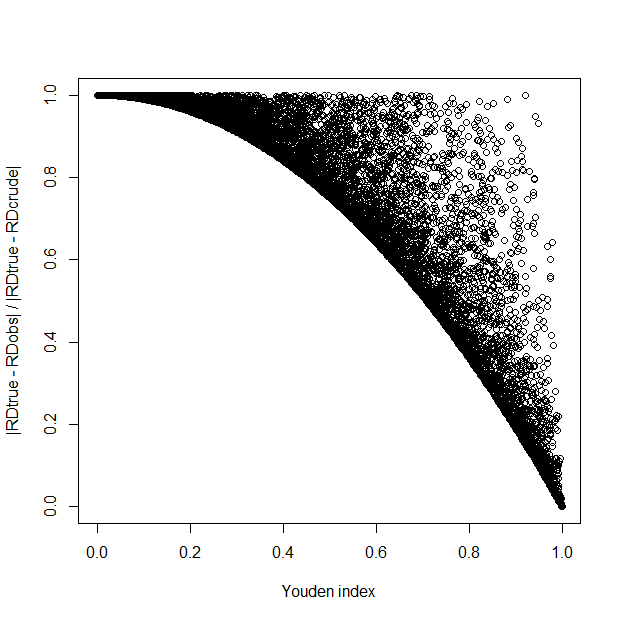}
\end{tabular}\caption{(tl) Histogram of the interval length. (tr) Distance between $RD_{obs}$ and $RD_{true}$ relative to the interval length. (bl) Zoom of the previous plot. (br) Distance between $RD_{obs}$ and $RD_{true}$ relative to the interval length, as a function of the strength of the dependence between $C$ and $D$ when measured by the Youden index.}\label{fig:plots}
\end{figure}

\subsection{Experiments}

In this section, we report some experiments that shed additional light on the relationships between the various risk differences under the conditions in Theorem \ref{the:nonmonotone}. For the experiments, we let $Y$ be binary. Then, we randomly parameterize 10000 times the causal graph to the left in Figure \ref{fig:graphs} by parameterizing the terms in the right-hand side of Equation \ref{eq:factorization} with parameter values drawn from a uniform distribution, while enforcing the assumptions in the theorem. For each parameterization, we compute $RD_{true}$, $RD_{obs}$ and $RD_{crude}$. Figure \ref{fig:plots} summarizes the results.\footnote{Code available at \texttt{https://www.dropbox.com/s/t75z8yro1e9higq/nonmonotonicity3.R?dl=0}.} The top left plot shows that most intervals are relatively small and, thus, that $RD_{obs}$ is close to $RD_{true}$ in most cases. However, the top right plot shows that $RD_{obs}$ tends to be closer to $RD_{crude}$ than to $RD_{true}$. The bottom left plot is a zoom of the previous plot at the smallest intervals. Finally, the bottom right plot shows that the stronger the dependence between $C$ and $D$ as measured by the Youden index (i.e., $p(d|c)+p(\nd|\nc)-1$), the closer $RD_{obs}$ is to $RD_{true}$. In summary, $RD_{obs}$ is a reasonable approximation to $RD_{true}$, but it is biased towards $RD_{crude}$. This may be a problem when the interval between $RD_{crude}$ and $RD_{true}$ is large. However, the length of the interval is unknown in practice, and we doubt substantive knowledge may provide hints on it. The bias decreases with increasing dependence between $C$ and $D$. Although the strength of this dependence is unknown in practice, substantive knowledge may give hints on it.

\section{Bounding the True Risk Difference}\label{sec:nonmonotone2}

Theorems and Corollaries \ref{the:Penna2020}-\ref{cor:nonmonotone2b} do not hold if the assumption that $p(a|c)=p(\na|\nc) \geq 0.5$ is replaced by the weaker assumption that $p(\na|\nc) \geq p(a|c) \geq 0.5$. Likewise for the assumption that $p(d|c)=p(\nd|\nc) \geq 0.5$. However, \citet[Theorems 5 and 6]{Penna2020} proves that the assumption $p(a|c)=p(\na|\nc) \geq 0.5$ can be relaxed and still $RD_{crude}$ and $RD_{obs}$ bound $RD_{true}$. We re-state this result in the next theorem.

\begin{theorem}[Pe\~na, 2020, Theorems 5 and 6]\label{the:Penna20202}
Consider the causal graph to the left in Figure \ref{fig:graphs}. Let $p(c)=0.5$, $p(d|c)=p(\nd|\nc) \geq 0.5$ and $p(\na|\nc) \geq p(a|c) \geq 0.5$. If $E[Y|a,c] - E[Y|a,\nc] \geq E[Y|\na,\nc] - E[Y|\na,c] \geq 0$, then $RD_{crude} \geq RD_{true}$ and $RD_{obs} \geq RD_{true}$. If $E[Y|a,c] - E[Y|a,\nc] \leq E[Y|\na,\nc] - E[Y|\na,c] \leq 0$, then $RD_{crude} \leq RD_{true}$ and $RD_{obs} \leq RD_{true}$.
\end{theorem}

Note that the previous theorem does not determine the order between $RD_{crude}$ and $RD_{obs}$. Thus, it cannot be used to decide which of the two comes closer to $RD_{true}$. However, the theorem may be useful to conclude whether $RD_{true}$ is positive or negative. For instance, the last result in the theorem allows us to conclude that $RD_{true} > 0$ whenever $\max(RD_{crude}, RD_{obs}) > 0$.

The following theorems show that the assumption that $p(d|c)=p(\nd|\nc) \geq 0.5$ in the previous theorem can also be relaxed.

\begin{theorem}\label{the:nonmonotone3}
Consider the causal graph to the left in Figure \ref{fig:graphs}. Let $p(c)=0.5$, $p(\na|\nc) \geq p(a|c) \geq 0.5$ and $p(\nd|\nc) \geq p(d|c) \geq 0.5$. If $E[Y|a,c] - E[Y|a,\nc] \geq E[Y|\na,\nc] - E[Y|\na,c] \geq 0$, then $RD_{crude} \geq RD_{true}$ and $RD_{obs} \geq RD_{true}$. If $E[Y|a,c] - E[Y|a,\nc] \leq E[Y|\na,\nc] - E[Y|\na,c] \leq 0$, then $RD_{crude} \leq RD_{true}$ and $RD_{obs} \leq RD_{true}$.
\end{theorem}

\begin{theorem}\label{the:nonmonotone4}
Consider the causal graph to the left in Figure \ref{fig:graphs}. Let $p(c)=0.5$, $p(a|c) \leq p(\na|\nc) \leq 0.5$ and $p(d|c) \leq p(\nd|\nc) \leq 0.5$. If $E[Y|a,c] - E[Y|a,\nc] \geq E[Y|\na,\nc] - E[Y|\na,c] \geq 0$, then $RD_{crude} \leq RD_{true}$ and $RD_{obs} \leq RD_{true}$. If $E[Y|a,c] - E[Y|a,\nc] \leq E[Y|\na,\nc] - E[Y|\na,c] \leq 0$, then $RD_{crude} \geq RD_{true}$ and $RD_{obs} \geq RD_{true}$.
\end{theorem}

Returning to our example of three diseases $A$, $D$ and $Y$ and a gene variant $C$, the assumptions for the first result in Theorem \ref{the:nonmonotone3} mean that (i) half of the population carry the gene variant $C$, i.e. $p(c)=0.5$, (ii) not carrying $C$ protects against $A$ and $D$ more than carrying it predisposes to suffer the diseases, i.e. $p(\na|\nc) \geq p(a|c) \geq 0.5$ and $p(\nd|\nc) \geq p(d|c) \geq 0.5$, and (iii) carrying $C$ increases the average severity of $Y$ for the individuals suffering $A$ more than it decreases the severity for the rest.

The last two theorems can be strengthened for $RD_{crude}$ as follows. Analogous results do not hold for $RD_{obs}$, though. 

\begin{theorem}\label{the:nonmonotone5}
Consider the causal graph to the left in Figure \ref{fig:graphs}. Let $p(c) \leq 0.5$ and $p(\na|\nc) \geq p(a|c) \geq 0.5$. If $E[Y|a,c] - E[Y|a,\nc] \geq E[Y|\na,\nc] - E[Y|\na,c] \geq 0$, then $RD_{crude} \geq RD_{true}$. If $E[Y|a,c] - E[Y|a,\nc] \leq E[Y|\na,\nc] - E[Y|\na,c] \leq 0$, then $RD_{crude} \leq RD_{true}$.
\end{theorem}

\begin{theorem}\label{the:nonmonotone6}
Consider the causal graph to the left in Figure \ref{fig:graphs}. Let $p(c) \geq 0.5$ and $p(a|c) \leq p(\na|\nc) \leq 0.5$. If $E[Y|a,c] - E[Y|a,\nc] \geq E[Y|\na,\nc] - E[Y|\na,c] \geq 0$, then $RD_{crude} \leq RD_{true}$. If $E[Y|a,c] - E[Y|a,\nc] \leq E[Y|\na,\nc] - E[Y|\na,c] \leq 0$, then $RD_{crude} \geq RD_{true}$.
\end{theorem}

\section{Path Diagrams}\label{sec:pathdiagrams}

Finally, we suppose that the variables $A$, $C$, $D$ and $Y$ are all continuous and follow the linear structural equation model represented by the path diagram to the right in Figure \ref{fig:graphs}. The true, crude and partially adjusted average causal effects of $A$ on $Y$ are given by the partial regression coefficients $\beta_{YA \cdot C}$, $\beta_{YA}$ and $\beta_{YA \cdot D}$, respectively. Note that the first cannot be computed because $C$ is unobserved. The following theorem proves that the partially adjusted average causal effect lies between the true and the crude ones and, thus, it comes closer to the true average causal effect than the crude.

\begin{theorem}\label{the:pathdiagram}
Consider the path diagram to the right in Figure \ref{fig:graphs}. Assume that the variables are standardized. If $sign(\beta)=sign(\gamma)$ then $\beta_{YA \cdot C} \leq \beta_{YA \cdot D} \leq \beta_{YA}$, else $\beta_{YA \cdot C} \geq \beta_{YA \cdot D} \geq \beta_{YA}$.
\end{theorem}

Note that unlike in the discrete case, no assumptions about the causal relationships of the variables are required to conclude that the partially adjusted average causal effect lies between the true and the crude ones. Note also that the signs of $\beta$ and $\gamma$ tell us whether the partially adjusted average causal effect is an upper or lower bound of the true one.

\section{Discussion}\label{sec:discussion}

One may think that adjusting for a proxy of a latent confounder is always a good idea. However, it is not. In this work, we have described sufficient conditions under which adjusting for a proxy of a latent confounder comes closer to the incomputable true average causal effect than not adjusting at all. Under some conditions, it is even possible to decide whether the partially adjusted approximation is an upper or a lower bound of the true quantity. We have experimentally shown that the partially adjusted approximation can be substantially better than the unadjusted one when the dependence between confounder and proxy is significant. We have also illustrated with an example that the conditions proposed are not too restrictive and unrealistic. Since one must rely on expert knowledge to verify the conditions, we would like to investigate in the future whether realistic, sufficient and empirically testable conditions exist. We would also like to extend this work to when several latent confounders exist.

\section*{Acknowledgments}

This work was funded by the Swedish Research Council (ref. 2019-00245).

\section*{Supplementary Material: Proofs}

\setcounter{theorem}{1}
\begin{theorem}
Consider the causal graph to the left in Figure \ref{fig:graphs}. Let $p(c)=0.5$, $p(a|c)=p(\na|\nc) \geq 0.5$ and $p(d|c)=p(\nd|\nc) \geq 0.5$. Then, $RD_{obs}$ lies between $RD_{true}$ and $RD_{crude}$.
\end{theorem}

\begin{proof}
We start by establishing a relationship between $RD_{obs}$ and $RD_{true}$. First, note that
\[
p(c|a,d) = \frac{p(a,d|c)p(c)}{p(a,d|c)p(c) + p(a,d|\nc)p(\nc)} = \frac{1}{1+\exp(-\delta(a,d))} = \sigma(\delta(a,d))
\]
where
\[
\delta(a,d) = \ln \frac{p(a,d|c)p(c)}{p(a,d|\nc)p(\nc)}
\]
is known as the log odds, and $\sigma()$ is known as the logistic sigmoid function \cite[Section 4.2]{Bishop2006}. Then,
\begin{equation}\label{eq:Bishop}
p(c|a,d) = \sigma \Big( \ln \frac{p(a,d|c)p(c)}{p(a,d|\nc)p(\nc)} \Big) = \sigma \Big( \ln \frac{p(a|c)p(d|c)}{p(a|\nc)p(d|\nc)} \Big)
\end{equation}
where the second equality follows from the assumption that $p(c)=0.5$, and the fact that $A$ and $D$ are conditionally independent given $C$ due to the causal graph under consideration. Likewise,
\begin{equation}\label{eq:Bishop2}
p(c|a,\nd) = \sigma \Big( \ln \frac{p(a|c)p(\nd|c)}{p(a|\nc)p(\nd|\nc)} \Big)
\end{equation}
and
\begin{equation}\label{eq:Bishop3}
p(c|d) = \sigma \Big( \ln \frac{p(d|c)}{p(d|\nc)} \Big)
\end{equation}
and
\begin{equation}\label{eq:Bishop4}
p(c|\nd) = \sigma \Big( \ln \frac{p(\nd|c)}{p(\nd|\nc)} \Big).
\end{equation}
Then, $p(c|a,d) \geq p(c|d)$ and $p(c|a,\nd) \geq p(c|\nd)$ because $\sigma()$ and $\ln()$ are increasing functions and $p(a|c) / p(a|\nc) \geq 1$ by the assumption that $p(a|c) = p(\na|\nc) \geq 0.5$. Then,
\begin{align}\label{eq:foo1}\nonumber
p(c|a,d) p(d) + p(c|a,\nd) p(\nd) &\geq p(c|d) p(d) + p(c|\nd) p(\nd)\\
&= \frac{p(d|c) p(c)}{p(d)} p(d) + \frac{p(\nd|c) p(c)}{p(\nd)} p(\nd) = 0.5
\end{align}
by the assumption that $p(c)=0.5$. Moreover,
\begin{equation}\label{eq:foo2}
p(\nc|a,d) p(d) + p(\nc|a,\nd) p(\nd) = 1 - ( p(c|a,d) p(d) + p(c|a,\nd) p(\nd) ).
\end{equation}

Next, note that
\begin{equation}\label{eq:equality1}
p(\nc|\na,\nd) = \sigma \Big( \ln \frac{p(\na|\nc)p(\nd|\nc)}{p(\na|c)p(\nd|c)} \Big) = \sigma \Big( \ln \frac{p(a|c)p(d|c)}{p(a|\nc)p(d|\nc)} \Big) = p(c|a,d)
\end{equation}
by the assumptions that $p(a|c)=p(\na|\nc)$ and $p(d|c)=p(\nd|\nc)$. Likewise,
\begin{equation}\label{eq:equality2}
p(\nc|\na,d) = \sigma \Big( \ln \frac{p(\na|\nc)p(d|\nc)}{p(\na|c)p(d|c)} \Big) = \sigma \Big( \ln \frac{p(a|c)p(\nd|c)}{p(a|\nc)p(\nd|\nc)} \Big) = p(c|a,\nd).
\end{equation}
Then,
\begin{equation}\label{eq:foo3}
p(c|a,d) p(d) + p(c|a,\nd) p(\nd) = p(\nc|\na,d) p(d) + p(\nc|\na,\nd) p(\nd)
\end{equation}
because
\begin{equation}\label{eq:pd}
p(d)=p(d|c)p(c)+p(d|\nc)p(\nc)=p(d|c)p(c)+p(\nd|c)p(\nc)=0.5.
\end{equation}
Moreover,
\begin{equation}\label{eq:foo4}
p(c|\na,d) p(d) + p(c|\na,\nd) p(\nd) = 1 - ( p(\nc|\na,d) p(d) + p(\nc|\na,\nd) p(\nd) ).
\end{equation}

Finally, Equations \ref{eq:foo1} and \ref{eq:foo3} allow us to write $p(c|a,d) p(d) + p(c|a,\nd) p(\nd) = p(\nc|\na,d) p(d) + p(\nc|\na,\nd) p(\nd) = 0.5 + \alpha$ with $\alpha \geq 0$, whereas Equations \ref{eq:foo2}, \ref{eq:foo3} and \ref{eq:foo4} allow us to write $p(\nc|a,d) p(d) + p(\nc|a,\nd) p(\nd) = p(c|\na,d) p(d) + p(c|\na,\nd) p(\nd) = 0.5 - \alpha$. Therefore,
\begin{align*}
RD_{obs} &= E[Y|a,d] p(d) + E[Y|a,\nd] p(\nd) - E[Y|\na,d] p(d) - E[Y|\na,\nd] p(\nd)\\
&= \big( E[Y|a,c,d]p(c|a,d) + E[Y|a,\nc,d]p(\nc|a,d) \big) p(d)\\
&+ \big( E[Y|a,c,\nd]p(c|a,\nd) + E[Y|a,\nc,\nd]p(\nc|a,\nd) \big) p(\nd)\\
&- \big( E[Y|\na,c,d]p(c|\na,d) + E[Y|\na,\nc,d]p(\nc|\na,d) \big) p(d)\\
&- \big( E[Y|\na,c,\nd]p(c|\na,\nd) + E[Y|\na,\nc,\nd]p(\nc|\na,\nd) \big) p(\nd)\\
& = E[Y|a,c] (p(c|a,d) p(d) + p(c|a,\nd) p(\nd))\\
& + E[Y|a,\nc] (p(\nc|a,d) p(d) + p(\nc|a,\nd) p(\nd))\\
& - E[Y|\na,c] (p(c|\na,d) p(d) + p(c|\na,\nd) p(\nd))\\
& - E[Y|\na,\nc] (p(\nc|\na,d) p(d) + p(\nc|\na,\nd) p(\nd))\\
& = E[Y|a,c] (0.5 + \alpha) + E[Y|a,\nc] (0.5 - \alpha)\\
& - E[Y|\na,c] (0.5 - \alpha) - E[Y|\na,\nc] (0.5 + \alpha)
\end{align*}
where the third equality follows from the fact that $Y$ and $D$ are conditionally independent given $A$ and $C$ due to the causal graph under consideration. Then,
\begin{equation}\label{eq:obstrue}
RD_{obs} = RD_{true} + \alpha \big( E[Y|a,c] - E[Y|a,\nc] + E[Y|\na,c] - E[Y|\na,\nc] \big)
\end{equation}
with $\alpha \geq 0$.

We continue by establishing a similar relationship between $RD_{obs}$ and $RD_{crude}$. First, note that
\begin{align}\nonumber\label{eq:x}
p(a|d)&=p(a|c,d)p(c|d)+p(a|\nc,d)p(\nc|d)\\
&=p(a|c)p(c|d)+p(a|\nc)p(\nc|d)\\\nonumber
&=p(\na|\nc)p(\nc|\nd)+p(\na|c)p(c|\nd)\\\nonumber
&=p(\na|\nc,\nd)p(\nc|\nd)+p(\na|c,\nd)p(c|\nd)\\\nonumber
&=p(\na|\nd)
\end{align}
by the fact that $A$ and $D$ are conditionally independent given $C$ due to the causal graph under consideration, the assumption that $p(a|c)=p(\na|\nc)$, and the fact that $p(c|d)=p(d|c)=p(\nd|\nc)=p(\nc|\nd)$ which follows from the assumptions that $p(c)=0.5$ and $p(d|c)=p(\nd|\nc)$ and the fact that $p(d)=0.5$ as shown in Equation \ref{eq:pd}. Likewise,
\begin{align}\nonumber\label{eq:y}
p(a|\nd)&=p(a|c,\nd)p(c|\nd)+p(a|\nc,\nd)p(\nc|\nd)\\
&=p(a|c)p(c|\nd)+p(a|\nc)p(\nc|\nd)\\\nonumber
&=p(\na|\nc)p(\nc|d)+p(\na|c)p(c|d)\\\nonumber
&=p(\na|\nc,d)p(\nc|d)+p(\na|c,d)p(c|d)\\\nonumber
&=p(\na|d).
\end{align}
Next, let $x=p(a|c)$ and $z=p(d|c)$. Recall that $p(a|c)=p(\na|\nc)$ by assumption, and $p(c|d)=p(d|c)=p(\nd|\nc)=p(\nc|\nd)$ as shown above. Then, Equation \ref{eq:y} can be rewritten as
\[
p(a|\nd)=x(1-z)+(1-x)z=-2xz+x+z.
\]
Recall that $z \geq 0.5$ by assumption. If $z=0.5$, then $-2xz+x+z=0.5$. On the other hand, if $z>0.5$ then $-2xz+x+z \leq 0.5$. To see it, assume to the contrary that
\begin{align*}
-2xz+x+z &>0.5\\
x(1-2z) &>0.5-z=(1-2z)/2\\
x &<0.5
\end{align*}
which contradicts the assumption that $x \geq 0.5$. Consequently, $p(a|\nd)=p(\na|d) \leq 0.5$ and $p(a|d)=p(\na|\nd) \geq 0.5$ by Equations \ref{eq:x} and \ref{eq:y}.

Finally, note that
\[
p(a)=p(a|c)p(c)+p(a|\nc)p(\nc)=p(a|c)p(c)+p(\na|c)p(\nc)=0.5
\]
by the assumptions that $p(a|c)=p(\na|\nc)$ and $p(c)=0.5$. This together with the fact that $p(d)=0.5$ as shown in Equation \ref{eq:pd}, and the previous paragraph imply that $p(\nd|a)=p(d|\na) \leq 0.5$ and $p(d|a)=p(\nd|\na) \geq 0.5$. Let $p(d|a)=p(\nd|\na)=0.5+\beta$ and $p(\nd|a)=p(d|\na)=0.5-\beta$ with $\beta \geq 0$. Then,
\begin{align}\nonumber
RD_{crude} &= E[Y|a] - E[Y|\na]\\\nonumber\label{eq:obscrude}
&= E[Y|a,d] p(d|a) + E[Y|a,\nd] p(\nd|a)\\\nonumber
&- E[Y|\na,d] p(d|\na) - E[Y|\na,\nd] p(\nd|\na)\\\nonumber
&= E[Y|a,d] (0.5 + \beta) + E[Y|a,\nd] (0.5 - \beta)\\\nonumber
&- E[Y|\na,d] (0.5 - \beta) - E[Y|\na,\nd] (0.5 + \beta)\\
&= RD_{obs} + \beta \big( E[Y|a,d] - E[Y|a,\nd] + E[Y|\na,d] - E[Y|\na,\nd] \big)
\end{align}
with $\beta \geq 0$, where the last equality follows from the fact that $p(d)=0.5$ by Equation \ref{eq:pd}.

Now, note that
\begin{align}\nonumber
E[Y|a,d] - E[Y|a,\nd] &= E[Y|a,c,d] p(c|a,d) + E[Y|a,\nc,d] p(\nc|a,d)\\\nonumber\label{eq:kk}
&- E[Y|a,c,\nd] p(c|a,\nd) - E[Y|a,\nc,\nd] p(\nc|a,\nd)\\\nonumber
&= E[Y|a,c] (p(c|a,d) - p(c|a,\nd))\\\nonumber
&- E[Y|a,\nc] (p(\nc|a,\nd) - p(\nc|a,d))\\\nonumber
&= E[Y|a,c] (p(c|a,d) - p(c|a,\nd))\\\nonumber
&- E[Y|a,\nc] (1 - p(c|a,\nd) - 1 + p(c|a,d))\\
&= \big( E[Y|a,c] - E[Y|a,\nc] \big) (p(c|a,d) - p(c|a,\nd))
\end{align}
where the second equality follows from the fact that $Y$ and $D$ are conditionally independent given $A$ and $C$ due to the causal graph under consideration. Likewise,
\begin{align}\nonumber
E[Y|\na,\nd] - E[Y|\na,d] &= E[Y|\na,c,\nd] p(c|\na,\nd) + E[Y|\na,\nc,\nd] p(\nc|\na,\nd)\\\nonumber\label{eq:kk2}
&- E[Y|\na,c,d] p(c|\na,d) - E[Y|\na,\nc,d] p(\nc|\na,d)\\\nonumber
&= E[Y|\na,\nc] (p(\nc|\na,\nd) - p(\nc|\na,d))\\\nonumber
&- E[Y|\na,c] (p(c|\na,d) - p(c|\na,\nd))\\\nonumber
&= E[Y|\na,\nc] (p(\nc|\na,\nd) - p(\nc|\na,d))\\\nonumber
&- E[Y|\na,c] (1 - p(\nc|\na,d) - 1 + p(\nc|\na,\nd))\\\nonumber
&= \big( E[Y|\na,\nc] - E[Y|\na,c] \big) (p(\nc|\na,\nd) - p(\nc|\na,d))\\
&= \big( E[Y|\na,\nc] - E[Y|\na,c] \big) (p(c|a,d) - p(c|a,\nd))
\end{align}
where the last equality follows from Equations \ref{eq:equality1} and \ref{eq:equality2}. Moreover, Equations \ref{eq:equality1} and \ref{eq:equality2} also imply that $p(c|a,d) \geq p(c|a,\nd)$, because $p(d|c)=p(\nd|\nc) \geq 0.5$ by assumption. Then,
\begin{align}\nonumber\label{eq:sign}
&sign\big(E[Y|a,c] - E[Y|a,\nc] + E[Y|\na,c] - E[Y|\na,\nc]\big) =\\
&sign\big(E[Y|a,d] - E[Y|a,\nd] + E[Y|\na,d] - E[Y|\na,\nd]\big).
\end{align}
This equation together with Equations \ref{eq:obstrue} and \ref{eq:obscrude} imply the desired result.
\end{proof}

\begin{theorem}
Consider the causal graph to the left in Figure \ref{fig:graphs}. Let $p(c)=0.5$, $p(a|c)=p(\na|\nc) \leq 0.5$ and $p(d|c)=p(\nd|\nc) \leq 0.5$. Then, $RD_{obs}$ lies between $RD_{true}$ and $RD_{crude}$.
\end{theorem}

\begin{proof}
Similar to the proof of Theorem \ref{the:nonmonotone}. Specifically, the assumption that $p(a|c) = p(\na|\nc) \leq 0.5$ implies that $p(a|c) / p(a|\nc) \leq 1$, which implies that $p(c|a,d) \leq p(c|d)$ and $p(c|a,\nd) \leq p(c|\nd)$, which implies that $p(c|a,d) p(d) + p(c|a,\nd) p(\nd) \leq 0.5$, which implies that
\begin{equation}\label{eq:obstrue2}
RD_{obs} = RD_{true} + \alpha \big( E[Y|a,c] - E[Y|a,\nc] + E[Y|\na,c] - E[Y|\na,\nc] \big)
\end{equation}
with $\alpha \leq 0$. Likewise, the assumption that $p(d|c) = p(\nd|\nc) \leq 0.5$ implies that $p(a|\nd)=p(\na|d) \leq 0.5$ and $p(a|d)=p(\na|\nd) \geq 0.5$, which implies that $p(\nd|a)=p(d|\na) \leq 0.5$ and $p(d|a)=p(\nd|\na) \geq 0.5$, which implies that
\begin{equation}\label{eq:obscrude2}
RD_{crude} = RD_{obs} + \beta \big( E[Y|a,d] - E[Y|a,\nd] + E[Y|\na,d] - E[Y|\na,\nd] \big)
\end{equation}
with $\beta \geq 0$. The assumption that $p(d|c) = p(\nd|\nc) \leq 0.5$ also implies that $p(c|a,d) \leq p(c|a,\nd)$, which implies that
\begin{align}\nonumber\label{eq:sign2}
&sign\big(E[Y|a,c] - E[Y|a,\nc] + E[Y|\na,c] - E[Y|\na,\nc]\big) =\\
&- sign\big(E[Y|a,d] - E[Y|a,\nd] + E[Y|\na,d] - E[Y|\na,\nd]\big).
\end{align}
This equation together with Equations \ref{eq:obstrue2} and \ref{eq:obscrude2} imply the desired result.
\end{proof}

\begin{corollary}
Consider the causal graph to the left in Figure \ref{fig:graphs}. Let $p(c)=0.5$, $p(a|c)=p(\na|\nc) \geq 0.5$ and $p(d|c)=p(\nd|\nc) \geq 0.5$. If $E[Y|a,c] - E[Y|a,\nc] \geq E[Y|\na,\nc] - E[Y|\na,c] \geq 0$, then $RD_{crude} \geq RD_{obs} \geq RD_{true}$. If $E[Y|a,c] - E[Y|a,\nc] \leq E[Y|\na,\nc] - E[Y|\na,c] \leq 0$, then $RD_{crude} \leq RD_{obs} \leq RD_{true}$.
\end{corollary}

\begin{proof}
It follows from Equations \ref{eq:obstrue}, \ref{eq:obscrude} and \ref{eq:sign}.
\end{proof}

\begin{corollary}
Consider the causal graph to the left in Figure \ref{fig:graphs}. Let $p(c)=0.5$, $p(a|c)=p(\na|\nc) \leq 0.5$ and $p(d|c)=p(\nd|\nc) \leq 0.5$. If $E[Y|a,c] - E[Y|a,\nc] \geq E[Y|\na,\nc] - E[Y|\na,c] \geq 0$, then $RD_{crude} \leq RD_{obs} \leq RD_{true}$. If $E[Y|a,c] - E[Y|a,\nc] \leq E[Y|\na,\nc] - E[Y|\na,c] \leq 0$, then $RD_{crude} \geq RD_{obs} \geq RD_{true}$.
\end{corollary}

\begin{proof}
It follows from Equations \ref{eq:obstrue2}, \ref{eq:obscrude2} and \ref{eq:sign2}.
\end{proof}

\begin{corollary}
Under the conditions in Corollary \ref{cor:nonmonotoneb}, $E[Y|a,c] - E[Y|a,\nc] \geq E[Y|\na,\nc] - E[Y|\na,c] \geq 0$ if and only if $E[Y|a,d] - E[Y|a,\nd] \geq E[Y|\na,\nd] - E[Y|\na,d] \geq 0$. Likewise when replacing $\geq$ with $\leq$.
\end{corollary}

\begin{proof}
It follows from Equations \ref{eq:kk} and \ref{eq:kk2}. Recall that $p(c|a,d) \geq p(c|a,\nd)$ was established in the proof of Theorem \ref{the:nonmonotone}.
\end{proof}

\begin{corollary}
Under the conditions in Corollary \ref{cor:nonmonotone2b}, $E[Y|a,c] - E[Y|a,\nc] \geq E[Y|\na,\nc] - E[Y|\na,c] \geq 0$ if and only if $E[Y|a,d] - E[Y|a,\nd] \leq E[Y|\na,\nd] - E[Y|\na,d] \leq 0$. Likewise when swapping $\leq$ and $\geq$.
\end{corollary}

\begin{proof}
It follows from Equations \ref{eq:kk} and \ref{eq:kk2}. Recall that $p(c|a,d) \leq p(c|a,\nd)$ was established in the proof of Theorem \ref{the:nonmonotone2}.
\end{proof}

\setcounter{theorem}{8}
\begin{theorem}
Consider the causal graph to the left in Figure \ref{fig:graphs}. Let $p(c)=0.5$, $p(\na|\nc) \geq p(a|c) \geq 0.5$ and $p(\nd|\nc) \geq p(d|c) \geq 0.5$. If $E[Y|a,c] - E[Y|a,\nc] \geq E[Y|\na,\nc] - E[Y|\na,c] \geq 0$, then $RD_{crude} \geq RD_{true}$ and $RD_{obs} \geq RD_{true}$. If $E[Y|a,c] - E[Y|a,\nc] \leq E[Y|\na,\nc] - E[Y|\na,c] \leq 0$, then $RD_{crude} \leq RD_{true}$ and $RD_{obs} \leq RD_{true}$.
\end{theorem}

\begin{proof}
We start by proving the first result in the theorem, specifically that $RD_{crude} \geq RD_{true}$. Recall from the proof of Theorem \ref{the:nonmonotone} that
\[
p(c|a) = \sigma \Big( \ln \frac{p(a|c)p(c)}{p(a|\nc)p(\nc)} \Big) = \sigma \Big( \ln \frac{p(a|c)}{p(a|\nc)} \Big)
\]
where the second equality follows from the assumption that $p(c)=0.5$. Likewise, 
\[
p(\nc|\na) = \sigma \Big( \ln \frac{p(\na|\nc)}{p(\na|c)} \Big).
\]
Therefore, $p(c|a) \geq 0.5$ and $p(\nc|\na) \geq 0.5$ due to the assumption that $p(\na|\nc) \geq p(a|c) \geq 0.5$. Now, consider the function $f(x)=x(1-x)$. By inspecting the first and second derivatives, we can conclude that $f(x)$ has a single maximum at $x=0.5$, and that it is increasing in the interval $[0,0.5]$ and decreasing in the interval $[0.5,1]$. This implies that $f(p(a|c))=p(a|c)p(\na|c) \geq p(a|\nc)p(\na|\nc)=f(p(\na|\nc))$ due to the assumption that $p(\na|\nc) \geq p(a|c) \geq 0.5$. Then,
\begin{equation}\label{eq:kuku}
\frac{p(a|c)}{p(a|\nc)} \geq \frac{p(\na|\nc)}{p(\na|c)}
\end{equation}
which together with the fact that $\sigma()$ and $\ln()$ are increasing functions imply that $p(c|a) \geq p(\nc|\na)$.

The results in the previous paragraph allow us to write $p(c|a)=0.5+\alpha$ and $p(\nc|\na)=0.5+\beta$ with $\alpha \geq \beta \geq 0$. Therefore,
\begin{align}\nonumber
RD_{crude} &= E[Y|a] - E[Y|\na]\\\nonumber
&= E[Y|a,c] p(c|a) + E[Y|a,\nc] p(\nc|a)\\\nonumber
&- E[Y|\na,c] p(c|\na) - E[Y|\na,\nc] p(\nc|\na)\\\nonumber
&= E[Y|a,c] (0.5 + \alpha) + E[Y|a,\nc] (0.5 - \alpha)\\\nonumber
&- E[Y|\na,c] (0.5 - \beta) - E[Y|\na,\nc] (0.5 + \beta)\\\label{eq:rdtrue2}
&= RD_{true} + \alpha \big( E[Y|a,c] - E[Y|a,\nc] \big) - \beta \big( E[Y|\na,\nc] - E[Y|\na,c] \big)
\end{align}
which implies that $RD_{crude} \geq RD_{true}$ because $\alpha \geq \beta \geq 0$, and $E[Y|a,c] - E[Y|a,\nc] \geq E[Y|\na,\nc] - E[Y|\na,c] \geq 0$ by assumption.

We continue by proving that $RD_{obs} \geq RD_{true}$. First, recall Equations \ref{eq:Bishop}-\ref{eq:Bishop4}. Then, $p(c|a,d) \geq p(c|d)$ and $p(c|a,\nd) \geq p(c|\nd)$ because $\sigma()$ and $\ln()$ are increasing functions and $p(a|c) / p(a|\nc) \geq 1$ by the assumption that $p(\na|\nc) \geq p(a|c) \geq 0.5$. Then,
\begin{align*}
p(c|a,d) p(d) + p(c|a,\nd) p(\nd) &\geq p(c|d) p(d) + p(c|\nd) p(\nd)\\
&= \frac{p(d|c) p(c)}{p(d)} p(d) + \frac{p(\nd|c) p(c)}{p(\nd)} p(\nd) = 0.5
\end{align*}
by the assumption that $p(c)=0.5$. We can analogously prove that $p(\nc|\na,d) p(d) + p(\nc|\na,\nd) p(\nd) \geq 0.5$. Moreover, it also holds that
\[
p(c|a,d) p(d) + p(c|a,\nd) p(\nd) \geq p(\nc|\na,d) p(d) + p(\nc|\na,\nd) p(\nd).
\]
To prove this inequality and after failing to do it on our own, we resorted to the function \texttt{FindInstance} from \texttt{Mathematica 12.2.0}. Specifically, we used \texttt{FindInstance} to find an instance of the probabilities that satisfied the reverse of the inequality above subject to $p(c)=0.5$, $p(\na|\nc) \geq p(a|c) \geq 0.5$ and $p(\nd|\nc) \geq p(d|c) \geq 0.5$. Since no such instance was found, the inequality above must hold. It is worth mentioning that \texttt{FindInstance} works analytically and not numerically and, thus, its outcome is exact and correct.\footnote{Code available at \texttt{https://www.dropbox.com/s/i8pnmm9zz4pqmgp/inequality.nb?dl=0}.}

The results in the previous paragraph allow us to write $p(c|a,d) p(d) + p(c|a,\nd) p(\nd) = 0.5 + \alpha$ and $p(\nc|\na,d) p(d) + p(\nc|\na,\nd) p(\nd) = 0.5 + \beta$ with $\alpha \geq \beta \geq 0$. Consequently, $p(\nc|a,d) p(d) + p(\nc|a,\nd) p(\nd) = 1 - ( p(c|a,d) p(d) + p(c|a,\nd) p(\nd) ) = 0.5 - \alpha$, and $p(c|\na,d) p(d) + p(c|\na,\nd) p(\nd) = 1- ( p(\nc|\na,d) p(d) + p(\nc|\na,\nd) p(\nd) ) = 0.5 - \beta$. Therefore,
\begin{align}\nonumber
RD_{obs} &= E[Y|a,d] p(d) + E[Y|a,\nd] p(\nd) - E[Y|\na,d] p(d) - E[Y|\na,\nd] p(\nd)\\\nonumber
&= \big( E[Y|a,c,d]p(c|a,d) + E[Y|a,\nc,d]p(\nc|a,d) \big) p(d)\\\nonumber
&+ \big( E[Y|a,c,\nd]p(c|a,\nd) + E[Y|a,\nc,\nd]p(\nc|a,\nd) \big) p(\nd)\\\nonumber
&- \big( E[Y|\na,c,d]p(c|\na,d) + E[Y|\na,\nc,d]p(\nc|\na,d) \big) p(d)\\\nonumber
&- \big( E[Y|\na,c,\nd]p(c|\na,\nd) + E[Y|\na,\nc,\nd]p(\nc|\na,\nd) \big) p(\nd)\\\nonumber
& = E[Y|a,c] (p(c|a,d) p(d) + p(c|a,\nd) p(\nd))\\\nonumber
& + E[Y|a,\nc] (p(\nc|a,d) p(d) + p(\nc|a,\nd) p(\nd))\\\nonumber
& - E[Y|\na,c] (p(c|\na,d) p(d) + p(c|\na,\nd) p(\nd))\\\nonumber
& - E[Y|\na,\nc] (p(\nc|\na,d) p(d) + p(\nc|\na,\nd) p(\nd))\\\nonumber
& = E[Y|a,c] (0.5 + \alpha) + E[Y|a,\nc] (0.5 - \alpha)\\\nonumber
& - E[Y|\na,c] (0.5 - \beta) - E[Y|\na,\nc] (0.5 + \beta)\\\label{eq:alphabeta}
& = RD_{true} + \alpha \big( E[Y|a,c] - E[Y|a,\nc] \big) - \beta \big( E[Y|\na,\nc] - E[Y|\na,c] \big)
\end{align}
where the third equality follows from the fact that $Y$ and $D$ are conditionally independent given $A$ and $C$ due to the causal graph under consideration. Then, $RD_{obs} \geq RD_{true}$ because $\alpha \geq \beta \geq 0$, and $E[Y|a,c] - E[Y|a,\nc] \geq E[Y|\na,\nc] - E[Y|\na,c] \geq 0$ by assumption.

Finally, the second result in the theorem follows Equations \ref{eq:rdtrue2} and \ref{eq:alphabeta}.
\end{proof}

\begin{theorem}
Consider the causal graph to the left in Figure \ref{fig:graphs}. Let $p(c)=0.5$, $p(a|c) \leq p(\na|\nc) \leq 0.5$ and $p(d|c) \leq p(\nd|\nc) \leq 0.5$. If $E[Y|a,c] - E[Y|a,\nc] \geq E[Y|\na,\nc] - E[Y|\na,c] \geq 0$, then $RD_{crude} \leq RD_{true}$ and $RD_{obs} \leq RD_{true}$. If $E[Y|a,c] - E[Y|a,\nc] \leq E[Y|\na,\nc] - E[Y|\na,c] \leq 0$, then $RD_{crude} \geq RD_{true}$ and $RD_{obs} \geq RD_{true}$.
\end{theorem}

\begin{proof}
Similar to the proof of Theorem \ref{the:nonmonotone3}. Specifically, the assumption that $p(a|c) \leq p(\na|\nc) \leq 0.5$ implies that $p(a|c)p(\na|c) \leq p(a|\nc)p(\na|\nc)$, which together imply that $p(c|a) \leq p(\nc|\na) \leq 0.5$, which implies that
\[
RD_{crude} = RD_{true} + \alpha \big( E[Y|a,c] - E[Y|a,\nc] \big) - \beta \big( E[Y|\na,\nc] - E[Y|\na,c] \big)
\]
with $\alpha \leq \beta \leq 0$. This implies the stated relationships between $RD_{crude}$ and $RD_{true}$. The assumption that $p(a|c) \leq p(\na|\nc) \leq 0.5$ also implies that $p(a|c) / p(a|\nc) \leq 1$, which implies that $p(c|a,d) \leq p(c|d)$ and $p(c|a,\nd) \leq p(c|\nd)$, which together with the assumption that $p(d|c) \leq p(\nd|\nc) \leq 0.5$ imply that
\[
p(c|a,d) p(d) + p(c|a,\nd) p(\nd) \leq p(\nc|\na,d) p(d) + p(\nc|\na,\nd) p(\nd) \leq 0.5
\]
which implies that
\[
RD_{obs} = RD_{true} + \alpha \big( E[Y|a,c] - E[Y|a,\nc] \big) - \beta \big( E[Y|\na,\nc] - E[Y|\na,c] \big)
\]
with $\alpha \leq \beta \leq 0$. This implies the stated relationships between $RD_{obs}$ and $RD_{true}$.
\end{proof}

\begin{theorem}
Consider the causal graph to the left in Figure \ref{fig:graphs}. Let $p(c) \leq 0.5$ and $p(\na|\nc) \geq p(a|c) \geq 0.5$. If $E[Y|a,c] - E[Y|a,\nc] \geq E[Y|\na,\nc] - E[Y|\na,c] \geq 0$, then $RD_{crude} \geq RD_{true}$. If $E[Y|a,c] - E[Y|a,\nc] \leq E[Y|\na,\nc] - E[Y|\na,c] \leq 0$, then $RD_{crude} \leq RD_{true}$.
\end{theorem}

\begin{proof}
Recall from the proof of Theorem \ref{the:nonmonotone} that
\[
p(c|a) = \sigma \Big( \ln \frac{p(a|c)p(c)}{p(a|\nc)p(\nc)} \Big)
\]
and
\[
p(\nc|\na) = \sigma \Big( \ln \frac{p(\na|\nc)p(\nc)}{p(\na|c)p(c)} \Big)
\]
and
\[
p(c) = \sigma \Big( \ln \frac{p(c)}{p(\nc)} \Big)
\]
and
\[
p(\nc) = \sigma \Big( \ln \frac{p(\nc)}{p(c)} \Big).
\]
Therefore, $p(c|a) \geq p(c)$ and $p(\nc|\na) \geq p(\nc)$ because $\sigma()$ and $\ln()$ are increasing functions and $p(a|c) / p(a|\nc) \geq 1$ and $p(\na|\nc) / p(\na|c) \geq 1$ by the assumption that $p(\na|\nc) \geq p(a|c) \geq 0.5$. Then, we can write $p(c|a)=p(c)+\alpha$ and $p(\nc|\na)=p(\nc)+\beta$ with $\alpha, \beta \geq 0$. Moreover, $\alpha \geq \beta$. To see it, recall from the proof of Theorem \ref{the:nonmonotone3} that a function of the form $f(x)=x(1-x)$ has a single maximum at $x=0.5$, and it is increasing in the interval $[0,0.5]$ and decreasing in the interval $[0.5,1]$. Now, note that $\sigma'(z) = \sigma(z) (1-\sigma(z))$ \citep[Equation 4.88]{Bishop2006}. Then, $\sigma'(z)$ has a single maximum at $\sigma(z)=0.5$ (i.e. at $z=0$), and it is increasing in the interval $\{\sigma(z) \:|\: 0 \leq \sigma(z) \leq 0.5\}$ (i.e., $\{z \:|\: -\infty < z \leq 0\}$) and decreasing in the interval $\{\sigma(z) \:|\: 0.5 \leq \sigma(z) \leq 1\}$ (i.e., $\{z \:|\: 0 \leq z < +\infty\}$). In other words, $\sigma(z)$ increases at an increasing rate in the interval $(-\infty,0]$ and increases at a decreasing rate in the interval $[0,+\infty)$. Therefore, $\sigma(-u+v) - \sigma(-u) \geq \sigma(u+v) - \sigma(u)$ for all $u,v \geq 0$. Then,
\begin{align*}
\alpha &= p(c|a) - p(c)\\
&= \sigma \Big( \ln \frac{p(a|c)}{p(a|\nc)} + \ln \frac{p(c)}{p(\nc)} \Big) - \sigma \Big( \ln \frac{p(c)}{p(\nc)} \Big)\\
&\geq \sigma \Big( \ln \frac{p(a|c)}{p(a|\nc)} + \ln \frac{p(\nc)}{p(c)} \Big) - \sigma \Big( \ln \frac{p(\nc)}{p(c)} \Big)\\
&\geq \sigma \Big( \ln \frac{p(\na|\nc)}{p(\na|c)} + \ln \frac{p(\nc)}{p(c)} \Big) - \sigma \Big( \ln \frac{p(\nc)}{p(c)} \Big)\\
&= p(\nc|\na) - p(\nc) = \beta
\end{align*}
where the first inequality follows from the fact that $\sigma(-u+v) - \sigma(-u) \geq \sigma(u+v) - \sigma(u)$ with $u=\ln \frac{p(\nc)}{p(c)}$ and $v=\ln \frac{p(a|c)}{p(a|\nc)}$, and the second inequality follows from Equation \ref{eq:kuku} and the fact that $\sigma()$ and $\ln()$ are increasing functions. Note that $u,v \geq 0$ by the assumptions that $p(c) \leq 0.5$ and $p(\na|\nc) \geq p(a|c) \geq 0.5$.

Finally, note that
\begin{align}\nonumber
RD_{crude} &= E[Y|a] - E[Y|\na]\\\nonumber
&= E[Y|a,c] p(c|a) + E[Y|a,\nc] p(\nc|a)\\\nonumber
&- E[Y|\na,c] p(c|\na) - E[Y|\na,\nc] p(\nc|\na)\\\nonumber
&= E[Y|a,c] (p(c) + \alpha) + E[Y|a,\nc] (p(c) - \alpha)\\\nonumber
&- E[Y|\na,c] (p(\nc) - \beta) - E[Y|\na,\nc] (p(\nc) + \beta)\\\label{eq:rdtruekuku}
&= RD_{true} + \alpha \big( E[Y|a,c] - E[Y|a,\nc] \big) - \beta \big( E[Y|\na,\nc] - E[Y|\na,c] \big)
\end{align}
which implies the desired results because $\alpha \geq \beta \geq 0$.
\end{proof}

\begin{theorem}
Consider the causal graph to the left in Figure \ref{fig:graphs}. Let $p(c) \geq 0.5$ and $p(a|c) \leq p(\na|\nc) \leq 0.5$. If $E[Y|a,c] - E[Y|a,\nc] \geq E[Y|\na,\nc] - E[Y|\na,c] \geq 0$, then $RD_{crude} \leq RD_{true}$. If $E[Y|a,c] - E[Y|a,\nc] \leq E[Y|\na,\nc] - E[Y|\na,c] \leq 0$, then $RD_{crude} \geq RD_{true}$.
\end{theorem}

\begin{proof}
Similar to the proof of Theorem \ref{the:nonmonotone5}. Specifically, the assumption that $p(a|c) \leq p(\na|\nc) \leq 0.5$ implies that $p(c|a) \leq p(c)$ and $p(\nc|\na) \leq p(\nc)$, which implies that $p(c|a)=p(c)+\alpha$ and $p(\nc|\na)=p(\nc)+\beta$ with $\alpha \leq \beta \leq 0$, because $\sigma(-u+v) - \sigma(-u) \leq \sigma(u+v) - \sigma(u)$ for all $u \geq 0$ and $v \leq 0$ and, moreover, the reverse inequality of Equation \ref{eq:kuku} holds now. Finally, Equation \ref{eq:rdtruekuku} implies the desired results.
\end{proof}

\begin{theorem}
Consider the path diagram to the right in Figure \ref{fig:graphs}. Assume that the variables are standardized. If $sign(\beta)=sign(\gamma)$ then $\beta_{YA \cdot C} \leq \beta_{YA \cdot D} \leq \beta_{YA}$, else $\beta_{YA \cdot C} \geq \beta_{YA \cdot D} \geq \beta_{YA}$.
\end{theorem}

\begin{proof}
\citet[Section 3.11]{Pearl2013} shows that $\beta_{YA \cdot C}=\alpha$, $\beta_{YA}=\alpha+\beta \gamma$ and
\begin{equation}\label{eq:betaYAD}
\beta_{YA \cdot D} = \alpha + \frac{\gamma \beta (1-\delta^2)}{1-\beta^2 \delta^2}.
\end{equation}
Note that the linear structural equation model corresponding to the path diagram under consideration implies that
\[
A = \beta C + \epsilon_A
\]
where $\epsilon_A$ is an error term that is independent of $C$ and, thus,
\[
var(A)=\beta^2 var(C) + var(\epsilon_A)
\]
where $var(A)=var(C)=1$ due to the assumption that the variables are standardized. This implies that $\beta^2 \leq 1$. Similarly $\delta^2 \leq 1$. Then, $1-\delta^2 \leq 1-\beta^2 \delta^2$ in Equation \ref{eq:betaYAD}. The result is now immediate.
\end{proof}

\bibliographystyle{plainnat}
\bibliography{monotonicityAddendum}

\begin{thebibliography}{8}
\providecommand{\natexlab}[1]{#1}
\providecommand{\url}[1]{\texttt{#1}}
\expandafter\ifx\csname urlstyle\endcsname\relax
  \providecommand{\doi}[1]{doi: #1}\else
  \providecommand{\doi}{doi: \begingroup \urlstyle{rm}\Url}\fi

\bibitem[Bishop(2006)]{Bishop2006}
C.~M. Bishop.
\newblock \emph{{Pattern Recognition and Machine Learning}}.
\newblock Springer, 2006.

\bibitem[Greenland(1980)]{Greenland1980}
S.~Greenland.
\newblock {The Effect of Misclassification in the Presence of Covariates}.
\newblock \emph{American Journal of Epidemiology}, 112\penalty0 (4):\penalty0
  564--569, 1980.

\bibitem[Miao et~al.(2018)Miao, Geng, and Tchetgen~Tchetgen]{Miaoetal.2018}
W.~Miao, Z.~Geng, and E.~J. Tchetgen~Tchetgen.
\newblock {Identifying Causal Effects with Proxy Variables of an Unmeasured
  Confounder}.
\newblock \emph{Biometrika}, 105\penalty0 (4):\penalty0 987--993, 2018.

\bibitem[Ogburn and VanderWeele(2012)]{OgburnandVanderWeele2012a}
E.~L. Ogburn and T.~J. VanderWeele.
\newblock {On the Nondifferential Misclassification of a Binary Confounder}.
\newblock \emph{Epidemiology}, 23\penalty0 (3):\penalty0 433--439, 2012.

\bibitem[Ogburn and VanderWeele(2013)]{OgburnandVanderWeele2013}
E.~L. Ogburn and T.~J. VanderWeele.
\newblock {Bias Attenuation Results for Nondifferentially Mismeasured Ordinal
  and Coarsened Confounders}.
\newblock \emph{Biometrika}, 100\penalty0 (1):\penalty0 241--248, 2013.

\bibitem[Pe\~na(2020)]{Penna2020}
J.~M. Pe\~na.
\newblock {On the Monotonicity of a Nondifferentially Mismeasured Binary
  Confounder}.
\newblock \emph{Journal of Causal Inference}, 8:\penalty0 150--163, 2020.

\bibitem[Pearl(2009)]{Pearl2009}
J.~Pearl.
\newblock \emph{Causality: Models, Reasoning, and Inference}.
\newblock Cambridge University Press, 2009.

\bibitem[Pearl(2013)]{Pearl2013}
J.~Pearl.
\newblock {Linear Models: A Useful ``Microscope'' for Causal Analysis}.
\newblock \emph{Journal of Causal Inference}, 1:\penalty0 155--170, 2013.

\end{thebibliography}

\end{document}